\documentclass[onecolumn]{IEEEtran}
\usepackage{graphicx}
\usepackage{color,colortbl,amssymb}
\usepackage{amsfonts,enumerate}
\usepackage{caption}
\usepackage{subcaption}
\usepackage{arydshln}

\usepackage{setspace}

\newtheorem{theorem}{Theorem}
\newtheorem{lemma}[theorem]{Lemma}

\newtheorem{corollary}[theorem]{Corollary}

\definecolor{LightCyan}{rgb}{0.88,1,1}

\newcommand{\cO}[1]{\mathcal{O}\left(#1\right)}
\newcommand{\gap}{{\delta}}

\newcommand{\graph}{{\cal G}}

\begin{document}

\title{GROTESQUE: Noisy Group Testing (Quick and Efficient)\thanks{GROTESQUE is short for {\bf GRO}up {\bf TES}ting {\bf QU}ick and {\bf E}fficient. }}

\author{\IEEEauthorblockN{Sheng Cai$^\ast$, Mohammad Jahangoshahi$^+$, Mayank Bakshi$^\ast$, and Sidharth Jaggi$^\ast$}\\
\IEEEauthorblockA{The Chinese University of Hong Kong$^\ast$, Sharif University of Technology$^+$}
}
\maketitle

\begin{abstract}
Group-testing refers to the problem of identifying (with high probability) a (small) subset of $D$ defectives from a (large) set of $N$ items via a ``small'' number of ``pooled'' tests ({\it i.e.}, tests that have a positive outcome if at least one of the items being tested in the pool is defective, else have a negative outcome). For ease of presentation in this work we focus on the regime when $D = \cO{N^{1-\gap}}$ for some $\gap > 0$.
The tests may be {\it noiseless} or {\it noisy}, and the testing procedure may be adaptive (the pool defining a test may depend on the outcome of a previous test), or non-adaptive (each test is performed independent of the outcome of other tests).
A rich body of literature demonstrates that $\Theta(D\log(N))$ tests are information-theoretically necessary and sufficient for the group-testing problem, and provides algorithms that achieve this performance.
However, it is only recently that reconstruction algorithms with computational complexity that is sub-linear in $N$ have started being investigated (recent work by~\cite{GurI:04,IndN:10, NgoP:11} gave some of the first such algorithms).
In the scenario with adaptive tests with noisy outcomes, we present the first scheme that is simultaneously order-optimal (up to small constant factors) in {\it both} the number of tests and the decoding complexity ($\cO{D\log(N)}$ in both the performance metrics). The total number of {\it stages} of our adaptive algorithm is ``small'' ($\cO{\log(D)}$). Similarly, in the scenario with non-adaptive tests with noisy outcomes, we present the first scheme that is simultaneously near-optimal in both the number of tests and the decoding complexity (via an algorithm that requires $\cO{D\log(D)\log(N)}$ tests and has a decoding complexity of {${\cal O}(D(\log N+\log^{2}D))$}.
Finally, we present an adaptive algorithm that only requires $2$ stages, and for which both the number of tests and the decoding complexity scale as {${\cal O}(D(\log N+\log^{2}D))$}.
For all three settings the probability of error of our algorithms scales as $\cO{1/(poly(D)}$.
\end{abstract}

%


\section{Introduction}
Let's say a ``large'' number (denoted $N$) of items contains a ``small'' number (denoted $D$, where $D$\footnote{In this work, we assume that the number of defective items is known {\it a priori}. If not, other work ({\it e.g., }\cite{SOBEL:75}) considers non-adaptive algorithms with low query complexity that help estimate $D$.} is assumed to be ``much smaller'' than $N$) of ``defective'' items.
The problem of {\it Group Testing} was first considered by Dorfman in 1943~\cite{Dor:43} as a means of identifying a small number of diseased individuals from a large population via as few ``pooled tests'' as possible. In this scenario, blood from a subset of individuals could be pooled together and tested in one go -- if the test outcome was ``negative'', then the subset {of} individuals tested in that test did not contain a diseased individual, otherwise it was ``positive''.
In the intervening decades a rich literature pertaining to the problem has built up and group testing has found many applications in different areas such as multiple access communication~\cite{GooH:08,Wol:1985}, DNA Library Screening~\cite{NgoD:2000,SchTR:2003,CheD:2008}-- a good survey of some of the algorithms, bounds and applications can be found in the books by Du and Hwang~\cite{DuH:00,DuH:1999}.

\noindent {\bf \underline{Number of tests:}}
A natural information-theoretic lower bound on the number of tests required to identify the set of defectives is $\Omega(D\log(N/D))$. One way of deriving this follows by noting that the number of bits required to even describe a subset of size $D$ from a set of size $N$ equals $\log {{N}\choose{D}} = D\log(N/D)(1+o(1))$, hence at least this many tests (each of which have binary outcomes) are needed.
In this work, for ease of presentation of our results, we focus on the setting where $D = \cO{N^{1-\gap}}$ for some $\gap > 0$ (though our results hold in greater generality). In this regime, the number of required tests scales as $\Omega({D\log(N)})$. Note that this argument also demonstrates that the decoding complexity of any group-testing algorithm scales as $\Omega({D\log(N)})$.\footnote{A slightly more involved information-theoretic argument is required if the group-testing procedure is allowed to fail with ``small'' probability. However, even in this setting, essentially the same lower bounds can be proven (up to small constant multiplicative factors that may depend on the allowed probability of error)  -- {\it e.g.},~\cite{Mal78, ChaCJS:11}.}

There are at least three different performance parameters for which the group-testing problem comes in different flavours, corresponding to whether the tests are noiseless or not, adaptive or not, and the algorithm is required to be zero-error or not.\footnote{Another flavour we do not focus on in this work corresponds to whether the design of the testing procedure is deterministic or randomized. Recent works -- for instance~\cite{PorR:08, Mazu:12} -- provide computationally efficient deterministic designs. In this work, however, we focus on ``Monte Carlo'' type randomized design algorithms.}
\begin{itemize}
\item {\it \underline{Noiseless vs. noisy tests}:} If test outcomes are always positive when tests contain at least one defective item, and always negative otherwise, then they are said to be {\it noiseless} tests. In some settings, however, the test outcomes are ``noisy'' -- {a} common model for this noise ({\it e.g.,}~\cite{MalS:98,ChaCJS:11,AtiS:12}) is when test outcomes are flipped i.i.d.\footnote{A ``worst-case'' noise model wherein {\it arbitrary} {(rather than {\it random})} errors, up to a constant fraction of tests, has also been considered in the literature ({\it e.g.,}~\cite{NgoP:11,AtiS:12}).} via Bernoulli($q$) noise.\footnote{Other types of noise have also been considered in the literature -- another well-analyzed type called ``dilution'' noise ({\it e.g.,}~\cite{Hwa:76}) corresponds to tests with fewer defectives having a higher chance of resulting in a false positive outcome.} It is known in several settings ({\it e.g.,}~\cite{MalS:98,ChaCJS:11,AtiS:12}) that the number of noisy tests required to reconstruct the set of defectives is at most a constant factor greater than the number of noiseless tests required (both requiring $\cO{D\log(N)}$ tests).
This constant factor depends on $q$ proportionally with $1/(1-H(q))$, where $H(q)$ is the binary entropy function. Hence in this work we focus on the more general setting, with noisy measurements.

\item {\it \underline{Adaptive vs. Non-adaptive tests}:} Whether the tests are noiseless and noisy tests, as long as a ``small'' probability of error is allowed for the reconstruction algorithm (Monte-Carlo algorithms), it turns out the number of tests required meets (up to a constant factor that may depend on $q$) the information-theoretic lower-bound of $\cO{D\log(N)}$ {({\it e.g.,}~\cite{BonGV:2005,MarT:2011,CheD:2008,DamM:2012,ChaCJS:11})}. Given this, non-adaptive algorithms are often preferred to adaptive algorithms in applications, since they allow for parallelizable implementation and/or the usage of off-the-shelf hardware. Even among the class of adaptive algorithms, it is preferable to have {\it as few} adaptive stages as possible. In this work we focus on both adaptive and non-adaptive algorithms. In the case of adaptive algorithms, we further also consider the case of adaptive algorithms with just two stages (one round of feedback).

\item \label{set:3} {\it \underline{Zero-error vs. ``small-error'' algorithms}:} Potential goals for group-testing algorithms (in the setting with noiseless tests) is to either {\it always} identify the set of defective items correctly, or to output it correctly ``with high probability''. With adaptive tests, it turns out both these settings require $\Theta{(D\log(N))}$ tests {({\it e.g.,}~\cite{CheD:2008} for the former setting and \cite{DamM:2012} for the latter)}. With non-adaptive tests, however, it turns out~\cite{DyaR:82, Dya:89} that requiring zero-error reconstruction {implies that at least $\Omega(D^2\log(N)/\log(D))$ tests must be performed}. Given this potentially large gap between the number of tests required in the two setting, in this work we focus on the ``small-error'' setting.
\end{itemize}

\noindent {\bf \underline{Decoding complexity:}}
The discussion above focused exclusively on the number of tests required, with no regard to the computational complexity of the corresponding reconstruction algorithms. While many of the algorithms reprised above are reasonably computationally efficient, the decoding complexity of most still scales at least linearly in $N$. Some notable exceptions are the results in~\cite{GurI:04,IndN:10,NgoP:11}. The most recent work in this line with decoding complexity that is sub-linear in $N$ culminated in a group-testing algorithm with $M = \cO{D^2\log(N)}$ tests, and decoding complexity that scales as $poly(M)$.\footnote{An algorithm with $\cO{D\log(N/D)}$ tests was also presented in the regime where $D = \Theta(N)$.} As noted in~\cite{IndN:10}, such algorithms can find applications in data-stream algorithms for problems such as the ``heavy hitter'' problem~\cite{CorM:05}
or in cryptographic applications such as the ``digital forensics'' problem~\cite{GooAT:05}.
Also as noted in~\cite{NgoP:11},  {\it ``[the group-testing] primitive has found many applications as stand alone objects and as building blocks in the construction of other combinatorial objects.''} We refer the reader interested in these and other applications to the excellent expositions in~\cite{IndN:10, NgoP:11}, and focus henceforth simply on the purely combinatorial problem of group-testing.

Our starting point is to note that since the sub-linear time algorithms in~\cite{IndN:10, NgoP:11} require zero-error reconstruction, the penalty paid in terms of the number of tests is heavy ($\Omega(D^2\log(N)/\log(D))$ as opposed to $\cO{D\log(N)}$.) Further, the decoding complexity is a low-degree polynomial in $M = \cO{D^2\log(N)}$, leaving a significant gap vis-a-vis the information-theoretic lower-bound of $\Omega(D\log(N))$ decoding steps (since any algorithm must examine at least $\cO{D\log(N)}$ test outcomes to have a ``reasonable'' probability of success).

\subsection{Our contributions}
In our work, we consider both the adaptive and non-adaptive group-testing settings, with noisy tests, and decoding error scaling as $\cO{1/poly(D)}$ in both settings.
\begin{itemize}
\item {\bf{\underline{\em Multi-stage adaptive algorithm}:}}
In the adaptive setting ours is the { first} algorithm  to be simultaneously order-optimal (up to a small constant factor that depends on the noise parameter $q$) in the number of tests required, and in decoding complexity -- $\Theta(D\log(N))$ for both measures. Our adaptive algorithm also does not need ``much'' adaptivity. In particular, our algorithm has $\cO{\log(D)}$ stages, where the tests within each stage are non-adaptive -- it is only across stages that adaptivity is required.

\item {\bf{\underline{\em Non-adaptive algorithm}:}} Analogously, in the non-adaptive setting
we present the {first} algorithm that is simultaneously near-optimal in both number of measurements and decoding complexity (requiring $\cO{D\log(D)\log(N)}$ tests and having a decoding complexity of {${\cal O}(D(\log N+\log^{2}D))$}.

\item {\bf{\underline{\em Two-stage adaptive algorithm}:}} Finally, combining ideas from the above, we present the first $2$-stage algorithm that is simultaneously near-optimal in both number of measurements and decoding complexity, with both the number of tests and the decoding complexity scaling as {${\cal O}(D(\log N+\log^{2}D))$}.
\end{itemize}

The rest of this paper is organized as follows. We first present the high-level overview of Grotesque tests (which is the main tool for our algorithm designs) and three group testing algorithms in Section \ref{sec:high-level}. Section \ref{sec:grotesque} to Section \ref{sec:2stage} contain detailed descriptions and analysis of Grotesque tests and our group testing algorithms. Section \ref{sec:conclusion} concludes this paper.

\section{High-level overview}\label{sec:high-level}
We now preview the key ideas used in designing our algorithms.

We begin by noting that our multi-round adaptive algorithm has decoding complexity that is information-theoretically order-optimal (and in some parameter ranges, for instance when $D = \cO{poly(\log (N))}$ the non-adaptive algorithm, and the $2$-round adaptive algorithm do too).
If our algorithms are to indeed be as blindingly fast as claimed above, it'd be very nice to have a black-box that has the following property -- with probability $1-\cO{1/poly(D)}$, given $\cO{\log(N)}$ (noisy) non-adaptive tests on a subset of items that contain exactly one defective item that has not yet been identified, in $\cO{\log(N)}$ time the black-box outputs the index number of this defective item. Our multi-stage adaptive group testing algorithm then gives to this black-box subsets of items that, with constant probability, contain exactly one unidentified defective item. Our non-adaptive group testing algorithm, on the other hand, gives to this black-box subsets of items that, with probability $\Omega(1/\log(D))$, contain exactly one unidentified defective item. These choices lead to the claimed performance of our algorithms.


\subsection{GROTESQUE Tests}

{We first describe a non-adaptive testing and decoding procedure (which we call GROTESQUE testing) that is useful in simulating such a black-box.

GROTESQUE first performs {\it multiplicity testing} -- it takes as inputs a set of $n$ items (where $n$ may in general be smaller than $N$), and ``quickly'' (in time $\cO{\log(D)}$) first estimates (with ``high'' probability) whether these $n$ items contain $0$, $1$, or more than one defectives. If the $n$ items contain $0$ or more than $1$ defectives, GROTESQUE outputs this information and terminates at this point. However, if the $n$ items contain exactly $1$ defective item, it then performs {\it localization} --  it ``quickly'' (in time $\cO{\log(n)}$) estimates (with ``high'' probability) the index number of this item. Both these processes (multiplicity testing, and localization) are non-adaptive.
\begin{itemize}
\item{\bf{\underline{\em Multiplicity tests}:}} The idea behind multiplicity testing is straightforward -- GROTESQUE simply performs $\Theta(\log(D))$ {\it random} tests, in which each of the $n$ items is present in each of the $\Theta(\log(D))$ tests with probability $1/2$ (hence these tests are non-adaptive). As Table~\ref{fig:table1} demonstrates, if the set of $n$ items being tested has exactly one defective item, then in expectation about half the $\Theta(\log(D))$ multiplicity tests should have positive outcomes, otherwise the number of tests with positive outcomes should be strictly bounded away from $1/2$ (even if the tests are noisy). In fact, the probability of error in the Multiplicity testing stage can be concentrated to be lower than $\exp({-\Theta(\log(D))}) = \cO{1/(poly(D))}$.

\item{\bf{\underline{\em Localization tests}:}} The idea behind {\it localization} is somewhat more involved. For this sub-procedure, GROTESQUE (non-adaptively) designs {\it a priori} a sequence of binary $\Theta({\log(N)})\times n$ matrices. In particular, the columns of each such matrix correspond to the collection of {\it all} codewords of a {\it constant-rate expander code}~\cite{Spi:95} with block-length $\Theta({\log(N)})$. In brief, these are error-correcting codes whose redundancy is a constant fraction of the block-length, and that can correct a constant fraction of bit-flips with ``high probability'' (for instance, Barg and Z\'emor~\cite{BarZ:04} analyze their performance against the ``probability $q$ bit-flip noise'' and demonstrate that the probability of error decays exponentially in the block-length). Further, expander codes have the very desirable property that their decoding complexity scales linearly in the block-length. But {conditioning on the event that the multiplicity of defectives in the $n$ items being tested equals exactly $1$} (say the $i$th item is defective), this means that in the {\it noiseless} setting, the binary vector of $\Theta({\log(N)})$ test outcomes corresponding to the localization tests performed by GROTESQUE correspond exactly to the $i$th codeword of the expander code. Even in the {\it noisy} setting, the vector of test outcomes corresponds to the $i$th codeword being corrupted by Bernoulli($q$) bit-flips. In both of these settings, by the guarantees provided in~\cite{Spi:95, BarZ:04}, the GROTESQUE localization procedure outputs the incorrect index (corresponding to the defective item) with probability $\exp(-\cO{\log(N)}) = \cO{1/(poly(N))} = o(1/(poly(D)))$. 

\end{itemize}

We now present the ideas behind our three algorithms, highlighting the use of GROTESQUE tests in each.
\subsection{Adaptive Group Testing}

For the adaptive group-testing problem, we now use a few ``classical'' combinatorial primitives (``balls and bins problem''/McDiarmid's concentration inequality/``coupon collector's problem''), combined carefully with the GROTESQUE testing procedure. We do this in two phases:

\begin{itemize}
\item{\bf\underline{\em Random binning}:}
We first note that if we randomly partition the set of all N items into say $2D$ disjoint pools (each with $N/2D$) items, then with ``high'' probability (via McDiarmid's inequality~\cite{Mcd:89}) a constant fraction of the pools contain exactly one defective item. Hence GROTESQUE can use the disjoint pools as inputs with $n = N/(2D)$, and in a single stage of $2D$ pools and corresponding $2D\cdot \cO{\log(n)} = \cO{D\log(N/D)}$ non-adaptive tests identify a constant fraction of the $D$ defective items (with probability at least $1-\exp(-\Theta(D))$). In the subsequent $\cO{\log(D)}$ stages, since the number of unidentified defectives decays geometrically, the number of pools per stage can be chosen to decay geometrically for comparable performance. Since the number of tests decay geometrically, the overall number of GROTESQUE tests sum up to $\cO{D}$. However, each GROTESQUE test requires at most $\log(N)$ tests with corresponding time-complexity $\cO{\log(N)}$. Hence the overall number of tests, and time-complexity, of these random binning stages is $\cO{D\log(N)}$.

However, by the time we're at the $\cO{\log(D)}$-${th}$ stage, the number of remaining unidentified defective items is ``small'' (at most $\log(D)$). Hence concentration inequalities may not provide the desired decay in the probability of error of that stage (corresponding to the event that the stage correctly recovers less than a certain constant fraction of the defective items remaining from the previous stage). The overall probability of error of all the random-binning stages is dominated by the probability of error of the last random-binning stage.

\item{\bf\underline{\em Coupon collection}:}
To compensate for this ``problem of small numbers'', {in the last stage} we segue to an alternative primitive, that of {\it coupon collection}~\cite{MotR:95}. We choose parameters so that at the beginning of this coupon-collection {stage}, there are less than ${\log(D)}$ unidentified defectives remaining. Rather than {\it partitioning} the set of items into pools as in the previous stages, in this stage we {\it independently} choose
$\cO{\log^{2}(D)\log\log(D)}$ pools (corresponding to the ``{coupons}'' in the coupon-collector's problem) -- note that $\cO{\log^{2}(D)\log\log(D)} = o(D\log(N))$, hence this coupon-collection stage does not change the overall number of tests required by more than a constant factor. Each pool is chosen to be of size so that with constant probability it contains one of the remaining $\cO{\log(D)}$ unidentified defectives. Each pool/coupon is given as an input to GROTESQUE.
By standard concentration inequalities {on the coupon collection} process, after $\cO{\log^{2}(D).\log(\log(D))}$ coupons have been collected, with probability $1-\cO{1/poly(D)}$ all the defectives are decoded.
\end{itemize}

\subsection{Non-adaptive Group Testing}

The critical difference between adaptive and non-adaptive group testing arises from the fact that
defective items that have already been identified cannot be removed from future tests. This means that if  we na\"ively use the adaptive procedure outlined above and optimize parameters, it results in an algorithm with $\cO{D\log(D)\log(N)}$ items and $\cO{D\log(D)\log(N)}$ decoding complexity.

Instead, we redesign our testing procedure to speed up the decoding complexity to $\cO{D(\log^2(D)+\log(N))}$ (though we still need $\cO{D\log(D)\log(N)}$ tests). In particular, we first non-adaptively choose a set of $\cO{\log(D)}$ random graphs
{$\graph_g$s} with the following properties --
each graph is bipartite, has $N$ nodes on the left (and is left-regular with left-degree $1$) and $\cO{D}$ nodes on the right. Each right node corresponds to a group of $\cO{\log(N)}$ non-adaptive (GROTESQUE) tests, for a total of $\cO{D\log(D)\log(N)}$ non-adaptive tests.

A node on the right of $\graph_{g}$ is said to be a ``leaf node with respect to $\graph_{g}$'' if the left-nodes connected to it contain exactly one defective item (or in other words, GROTESQUE's multiplicity test, run on the items connected to such a node, would with high probability return a value of $1$).
It can be shown {via} standard concentration inequalities that for each defective item (on the left of each bipartite graph $\graph_{g}$), a constant fraction of its $\cO{\log(D)}$ right neighbours (over all $\cO{\log(D)}$ graphs) are such that they are ``leaf nodes with respect to $\graph_{g}$''.  {For each $\graph_{g}$, the items/left-nodes connected to its right nodes} may now be given as an input to GROTESQUE (with $n = \cO{\log(N)}$ tests (however, in our actual algorithm, not all right nodes of all $\graph_{g}$s are necessarily chosen as inputs to GROTESQUE -- the speedup in decoding complexity arises crucially from a more careful procedure in deciding which right nodes to use to give inputs to the GROTESQUE testing procedure).

{
Decoding proceeds by iteratively following the steps below:
\begin{enumerate}
\item \label{step:init} {\it Initialization of leaf-nodes:} We initialize a {\it leaf-node list} that contains all leaf nodes. We do this by picking right nodes and sequentially feeding the corresponding left-nodes to GROTESQUE's multiplicity testing procedure ({\it not} its localization procedure, at least yet).
\item \label{step:local} {\it Localization of a single defective item:} We pick a right node in the leaf node list, and use GROTESQUE's localization testing procedure on this node to identify the corresponding defective item.
\item \label{step:shrink} {\it Updating the leaf-node list:} We remove all the right neighbours of the defective item identified in the previous stage from each of the $\cO{\log(D)}$ graphs, and updating the leaf node list. Finally we return to Step~\ref{step:local}, until all $D$ defectives have been found.
\end{enumerate}
\noindent It can be verified that the first and third steps of this algorithm both take $\cO{D\log^2(D)}$ steps, and the second step takes $\cO{D\log(N)}$ steps, to give the overall desired computational complexity.}

\subsection{Two-stage Adaptive Group Testing}
{We now merge ideas from our previous algorithms to present an adaptive group testing algorithm with ``minimal adaptivity'' (just two stages). We also use in our algorithm a key primitive suggested in  Theorem 1 of~\cite{DamM:2012}, specifically ``birthday paradox hashing''.}

The main difference between our algorithm and the one presented in Theorem 1 of~\cite{DamM:2012} is that our algorithm is robust to Bernoulli($q$) noise, has decoding complexity scaling as ${\cal O}(D(\log N+\log^{2}D))$, and number of tests scaling as ${\cal O}(D(\log N+\log^{2}D))$. In contrast, the algorithm in~\cite{DamM:2012} requires fewer tests (${\cal O}(D\log N)$), but significantly higher decoding complexity ($\cO{\exp(N)}$), and is not robust to noise in the measurement process.

The high-level intuition behind the algorithm in Theorem 1 of~\cite{DamM:2012} is to first partition the $N$ items into at least $D^2$ groups. The ``Birthday Paradox'' is a simple calculation that demonstrate that if $D$ balls are thrown uniformly at random into more than $D^2$ bins, then the probability of a ``collision'' (there being a bin with more than one ball in it) is small.

Using this primitive, it follows that with high probability each group contains at most one defective item. In the first stage, $\cO{D\log (D^2)}$ non-adaptive group tests are performed to identify the $D$ groups (out of $D^2$) that contain exactly one defective.

In the second stage (that depends adaptively on the outcomes of the first stage), $\cO{\log{N/D^2}}$ non-adaptive group tests are performed on the $N/D^2$ items of each group that has been identified as containing a defective in the first stage. Thus, the total number of tests required for the second stage is $\cO{D\log{(N/D^2)}}$.

However, the high decoding complexity of the algorithm in Theorem 1 of~\cite{DamM:2012} arises from the fact that the non-adaptive group testing algorithm used has high decoding complexity. We hence substitute the non-adaptive group test used in their scheme with the one presented in Section~\ref{sec:nonadaptive} resulting in a drastic decrease in the decoding complexity at the cost of a potential slight increase in the number of tests required. Another relatively minor difference in our algorithm is that to get the probability of error to decay as $\cO{1/poly(D))}$ as desired for all our algorithms, we use $poly(D)$ groups (where the polynomial is of degree at least $3$) in the first stage instead of the $\Omega(D^2)$ groups used in the first stage of~\cite{DamM:2012}.

\begin{table*}[ht!]
\centering{}

\begin{tabular}{l|l}
\hline
\multicolumn{2}{l}{Group Testing}\tabularnewline
\hline
$N$ & The total number of items\tabularnewline
$D$ & The total number of defective items\tabularnewline
$\delta$ & $D={\cal O}\left(N^{1-\delta}\right)$\tabularnewline
$q$ & The pre-specified probability that the result of a test differs from
the true result\tabularnewline
${\cal D}$ & The set of all defective items\tabularnewline
${\cal N}$ & The set of all $N$ items \tabularnewline
$M$ & The total number of tests required to identify the set of defective
items\tabularnewline
\hline
\end{tabular}
\caption{Table of notation used for the general group testing problem}

\end{table*}

\begin{table*}[ht!]
\centering{}

\begin{tabular}{l|l}
\hline
\multicolumn{2}{l}{GROTESQUE TESTS}\tabularnewline
\hline
${\cal S}$ & The set of items being tested in GROTESQUE TESTS.\tabularnewline
$n$ & The number of items being tested in GROTESQUE TESTS, $n=|{\cal S}|$\tabularnewline
$d$ & The total number of defectives of GROTESQUE TESTS input\tabularnewline
$m_{1}$ & The number of Multiplicity tests\tabularnewline
$m_{2}$ & The number of Localization tests\tabularnewline
$K$ & The number of positive results of Multiplicity tests\tabularnewline
${\cal C}$ & Expander code\tabularnewline
${\bf y}^{(M)}$ & The length-$m_{1}$ binary vector $\left(y_{1}^{(M)},y_{2}^{(M)},\ldots,y_{m_{1}}^{(M)}\right)^{T}$denoting
the outcomes \tabularnewline
 & of the Multiplicity Encoder in the absence of noise.\tabularnewline
${\bf y}^{(L)}$ & The length-$m_{2}$ binary vector $\left(y_{1}^{(L)},y_{2}^{(L)},\ldots,y_{m_{2}}^{(L)}\right)^{T}$denoting
the outcomes \tabularnewline
 & of the Localization Encoder in the absence of noise.\tabularnewline
${\bf \hat{y}}^{(M)}$ & The length-$m_{1}$ binary vector denoting the actually observed noisy
outcomes \tabularnewline
 & of the Multiplicity Encoder.\tabularnewline
${\bf \hat{y}}^{(L)}$ & The length-$m_{2}$ binary vector denoting the actually observed noisy
outcomes \tabularnewline
 & of the Localization Encoder.\tabularnewline
\hline
\end{tabular}
\caption{Table of notation used in GROTESQUE tests}

\end{table*}

\section{GROTESQUE Tests}\label{sec:grotesque}

A key component of the algorithms that we present in this paper are GROTESQUE Tests (short for  {\bf GRO}up {\bf TES}ting {\bf QU}ick and {\bf E}fficient). Given a set  $\mathcal{S}=\{j_1,j_2,\ldots,j_n\}\subseteq{\mathcal{N}}$ containing an unknown number $d$ of defectives, the GROTESQUE tests tell us, with high probability, the number of defective items $d$ and the location if there is just one. The input to GROTESQUE tests is the $n$-length vector, $(x_{j}:j\in\mathcal{S})$, where $x_j$ is $1$ if $j$ is defective, and $0$ otherwise. The test outputs are $y_1,y_2,\ldots,y_m$, each of which are flipped independently by a Binary Symmetric Channel with transition probability $q$ to obtain noisy tests $\hat{y}_1,\ldots\hat{y}_m$. The noisy tests are then processed by the GROTESQUE decoder to output one of the following possibilities:
\begin{enumerate}
\item $d=0$, {\em i.e.}, there is no defective.
\item $d=1$. In this case, the decoder also outputs the location of the defective in set $\mathcal{S}$.
\item $d>1$, {\em i.e.}, there are at least two defective items.
\end{enumerate}
GROTESQUE consists of two kinds of tests - {\em multiplicity tests} and {\em localization tests}. Multiplicity tests tell us which of the above three possibilities it is, and the localization tests tell us which item is defective if there exists exactly one defective item in the $n$ items.

\begin{figure*}[ht]
\centering\label{fig:grotesque}
\includegraphics[width=\linewidth]{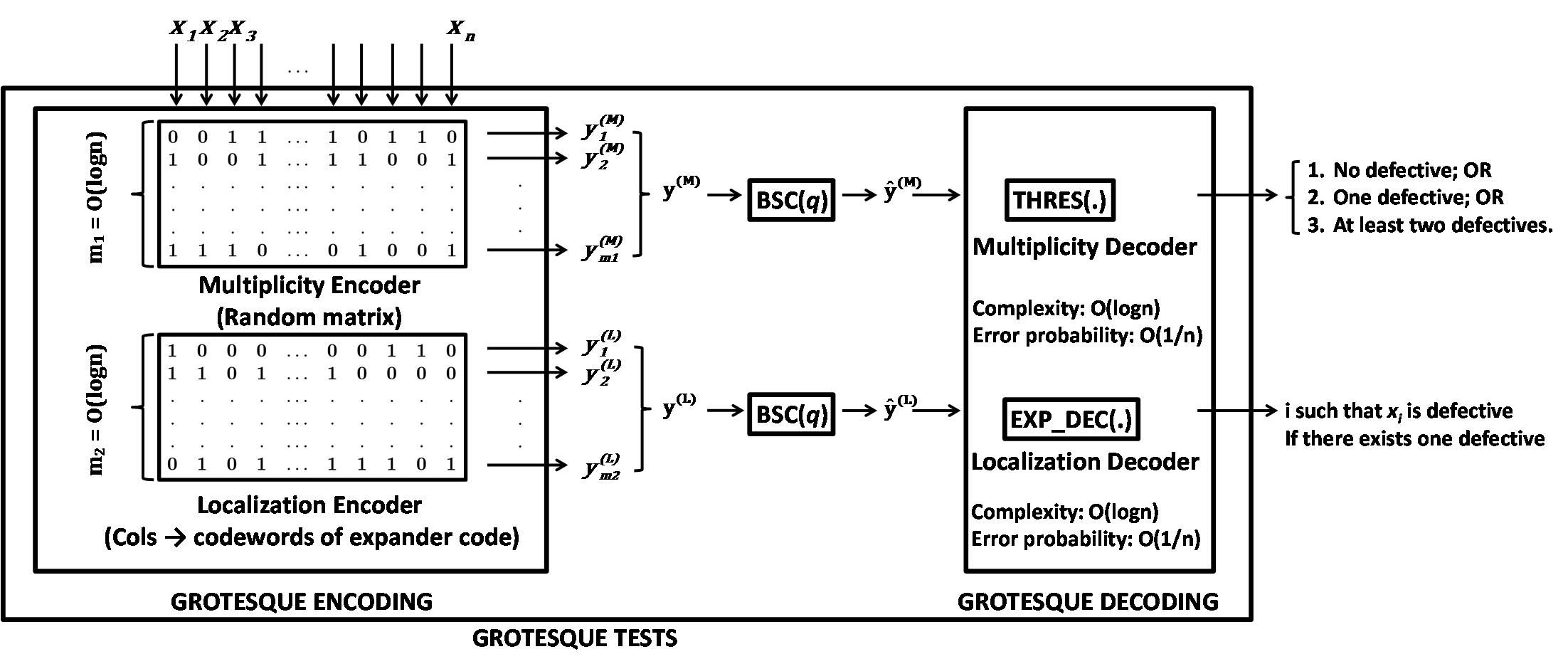}\caption{The input of a GROTESQUE tests is a set of $n$ items with unknown number
of defectives. There are three possible outputs: there exists no
defective item, or there exists exactly one defective item and the corresponding index number, or there exist at least two defective items (but GROTESQUE does not output the corresponding index numbers). There are
two parts to GROTESQUE Encoding: the Multiplicity Encoder and the Localization
Encoder. For the Multiplicity Encoder, we generate $m_{1}$ tests. Each
item is included in each test uniformly at random. In the Localization Encoder, whether {or not} a certain item is included in $m_{2}$ tests is determined
by the corresponding codeword of an expander code. The inputs to GROTESQUE
Decoding are the results of outputs of GROTESQUE Encoding passing
through $BSC(q)$. Again, GROTESQUE Decoding is divided into two parts:
the Multiplicity Decoder and the Localization Decoder. In the Decoder, we count the number of positive outputs of the Multiplicity Encoder,
compare {this number} with the expected numbers for three cases, and {decide on the multiplicity according to the rule given in Equation~\ref{eq:1}}. We call the above process $THRES(.)$ (short for {\it threshold detector}).
To implement the Localization
Decoder, we use $EXP\mbox{\_}DEC$ (short for {\it expander code decoder})
to do the decoding. If there exists exactly one defective item, the
output should be one of the codewords of the expander code. This tells us which item is defective.
The overall time complexity and error probability for the GROTESQUE tests are, respectively, ${\cal O}(\log n)$
and ${\cal O}(1/poly(n))$.}
\end{figure*}

\subsection{Multiplicity testing}

We generate $m_{1}$ tests in this part. In each test, the $j$-th
item is included with probability $1/2$. If we represent the tests
as a $m_{1}\times n$ matrix, $A^{(M)}$, then {each} entry of the $A^{(M)}$ is a Bernoulli random variable with parameter $1/2$. We count the number of positive items in a multiplicity test and denote it by $K$. We then use Equation~\ref{eq:1} to estimate the multiplicity value of the test. The use of Equation~\ref{eq:1} is justified by the values of Table~\ref{fig:table1}, which computes the expected value of the number of positive test outcomes $K$.

\begin{table}[tbh]
\centering{}
\begin{tabular}{|c||c|c|}
\hline
d & $K$ (Noiseless tests) & $K$ (Noisy tests) \tabularnewline
\hline\hline
$0$ & $0$ & $q\times m_{1}$\tabularnewline
\hline
$1$ & $1/2\times m_{1}$ & $1/2\times m_{1}$\tabularnewline
\hline
$\geq2$ & $\geq3/4\times m_{1}$ & $\geq(3/4-q/2)\times m_{1}$\tabularnewline
\hline
\end{tabular}
\caption{Expected number of positive test outcomes of the multiplicity tests.\label{fig:table1}
}
\end{table}

When the tests are noisy (as mentioned before, we assume that the noise follows the output of a BSC($q$)), GROTESQUE's multiplicity decoder uses the following rule to produce an output as to the multiplicity (when the tests are noiseless, the same equation with $q$ set to zero may be used):



\begin{equation}\label{eq:1}
d=\left\{ \begin{array}{cc}
0, & \mbox{if }K\in\left[0,\left(\frac{1}{4}+\frac{q}{2}\right)m_{1}\right)\\
1, & \mbox{if }K\in\left[\left(\frac{1}{4}+\frac{q}{2}\right)m_{1},\left(\frac{5}{8}-\frac{q}{4}\right)m_{1}\right)\\
\geq2, & \mbox{if }K\in\left[\left(\frac{5}{8}-\frac{q}{4}\right)m_{1},\infty\right)
\end{array}\right.
\end{equation}

\subsection{Localization}
{If the multiplicity test estimates $d$ to be $1$, we then use the results of the $m_{2}$ localization tests (which have been non-adaptively designed {\it a priori}) to localize the defective item.}
We represent the tests as a $m_{2}\times n$ matrix,
$A^{(L)}$. The difference {between $A^{(L)}$ and $A^{(M)}$ is that the columns of $A^{(L)}$ correspond to distinct
codewords} of an expander code, $\cal C$ (while the entries of $A^{(M)}$ were chosen uniformly at random). Different columns correspond to different
codewords. If there is exactly one defective item, then the output
 {of the localization tests} should be one of the codewords of $\cal C$ in the {scenario with noiseless tests}, or the result
of one of the codewords of $\cal C$ XOR'd with a vector  whose entries are
i.i.d. Bernoulli($q$) random variables in the {scenario with noisy tests}.
By Theorem \ref{thm:expandercode}, the Localization step is correct with error probability $\leq2^{-fm_{2}}$ (where $f$ is a constant for the code $\cal C$) and decoding complexity $\mathcal{O}(m_{2})$.

The following theorem about error-exponents of expander codes is useful in our construction.
\begin{theorem}[\cite{Barz:02}]
\label{thm:expandercode}For a given rate $R$, any
$\varepsilon>0$, and $\alpha<1$ there exists a polynomial-time constructible
code ${\cal C}$ of length $m_2$ such that $P_{e}({\cal C},\mbox{ }q)\leq2^{-\alpha m_2f(R,\mbox{ }q)}$,
where


\[
f(R,\mbox{ }q)={\max_{R\leq R_{0}\leq1-H(q)}}E(R_{0},\mbox{ }q)(R_{0}-R)/2-\varepsilon
\]

Here $E(R_{0},\mbox{ }q)$ is the ``random coding exponent'' and $H(\cdot)$ is the binary entropy function. The decoding complexity of a sequential implementation of this decoding
is $\mathcal{O}(m_2)$ and $f(R,\mbox{ }q)$ is positive\footnote{In \cite{BarZ:04}, an improvement of this result is given without the constraints on $q$.} for all $0<q<1-H(q)$.
\end{theorem}

\subsection{Performance Analysis}

In this section, we analyze the error probability and time complexity
of GROTESQUE tests in terms of $m_{1}$ and $m_{2}$. {The actual choices of these parameters are made in Sections~\ref{sec:adaptive},~\ref{sec:nonadaptive} and~\ref{sec:2stage}, corresponding respectively to  our adaptive, non-adaptive, and two-stage algorithms.}

\noindent $a)$ {\it Error probability:}

We now bound from above the probability of error of the multiplicity and localization sub-routines of GROTESQUE testing. In Lemma~\ref{lem:errorgrotesque} below, we explicitly derive the dependence of the probability of error of GROTESQUE tests on the value of $q$. This dependence can be directly translated into a dependence on the probability of error in each of our algorithms, but for ease of presentation we omit this dependence on $q$ outside this lemma (and focus only on the dependence on $D$ and $N$).
\begin{lemma}\label{lem:errorgrotesque}
\begin{itemize}
\item The error probability of GROTESQUE multiplicity testing is at most $\exp \left (-m_{1}(1-2q)^{2}/32 \right )$.
\item Conditioned on $d$ being correctly identified as $1$, the error probability of GROTESQUE multiplicity testing is at most $\exp(-\alpha m_{2}(1-H(q))^{3}/128)$, for some universal $\alpha > 0$.
\end{itemize}
\end{lemma}
\begin{proof}

\noindent {\it Probability of error for multiplicity testing:}
The outputs of each individual test in the multiplicity testing are i.i.d. Bernoulli random variables, with  mean depending on the value of $d$ and $q$.
There are three possible error events for multiplicity testing.
\begin{enumerate}
\item The true value of $d$ equals $0$, but GROTESQUE estimates it to be $\geq 1$. In this scenario the expected number of positive outcomes is $qm_1$. To decide on the value of $d$ (as either $0$ or $1$) the threshold that the multiplicity tester (given in Equation \ref{eq:1} is $\frac{1}{2}\left ( q + \frac{1}{2}\right )m_1$. Hence the multiplicity tester makes an error if the true number of positive outcomes exceeds the expected number by $x_1 = \frac{1}{2}\left ( q + \frac{1}{2}\right )m_1 - qm_1 = \frac{m_1}{2}\left ( \frac{1}{2} - q\right )$.
\item The true value of $d$ equals $1$, but GROTESQUE estimates it to be either $0$ or at least $2$. In this scenario the expected number of positive outcomes is $m_1/2$. To decide on the value of $d$ (as either $1$ or not) the closest threshold that the multiplicity tester (given in (\ref{eq:1}) is $\left (\frac{5}{8} -  \frac{q}{4}\right )m_1$. Hence the multiplicity tester may make an error if the true number of positive outcomes differs from the expected number by $x_2 = \left (\frac{5}{8} -  \frac{q}{4}\right )m_1 - \frac{m_1}{2}= \frac{m_1}{4}\left ( \frac{1}{2} - q\right )$.
\item The true value of $d$ is greater than or equal to $2$, but GROTESQUE estimates it to be either $0$ or $1$. In this scenario the expected number of positive outcomes is at least $\left (\frac{3}{4} -  \frac{q}{2}\right )m_1$. To decide on the value of $d$ (as either $\geq 2$ or not) the closest threshold that the multiplicity tester (given in (\ref{eq:1}) is $\left (\frac{5}{8} -  \frac{q}{4}\right )m_1$. Hence the multiplicity tester may make an error if the true number of positive outcomes differs from the expected number by $x_3 = \left (\frac{3}{4} -  \frac{q}{2}\right )m_1 - \left (\frac{5}{8} -  \frac{q}{4}\right )m_1 = \frac{m_1}{4}\left ( \frac{1}{2} - q\right )$.
\end{enumerate}

We now use the additive form of the Chernoff bound, which states that the probability of a $m_1$ i.i.d. copies of a binary random variable differing its expected value by more than $x$ is at most $2^{-2x^2/m_1}$. Noting that $x_1 > x_2 = x_3$ in the three cases analyzed, we have the desired bound on the probability of error.

\vspace{0.15in}
\noindent {\it Probability of error for localization testing:}
This error event corresponds to the scenario when there is exactly one defective item and
we claim that $d=1$, but we locate the wrong defective item. Here we use the exponentially small outer bound on the probability of error of expander codes, as shown in Theorem~\ref{thm:expandercode}, to bound our probability of error. The probability of error of the expander code is at most


\begin{eqnarray}\nonumber
P_{e}({\cal C},q) & \leq & 2^{-\alpha m_{2}f(R,q)}\\\nonumber
 & = & 2^{-\alpha m_{2}\max_{R\leq R_{0}\leq1-H(q)}E(R_{0},q)(R_{0}-R)/2-\epsilon}\\
 & \leq & 2^{-\alpha m_{2}(1-H(q)-R_{0})^{2}(1-H(q)-R)/4}\\\nonumber
 & = & 2^{-\alpha m_{2}(1-H(q)-R)^{3}/16}\nonumber\\
 & \leq & 2^{-\alpha m_{2}(1-H(q))^{3}/128}\nonumber,
\end{eqnarray}

\noindent where for the middle inequality we choose $R_{0}=\left(R+1-H(q)\right)/2$, and
the last inequality follows by choosing the rate $R$ of our expander code to never be greater than $(1-H(q))/2$.

Based on the above analysis, setting $m_{1}={\cal O}(\log D)$
and {$m_{2}={\cal O}(\log N)$}, we can guarantee that the error probability
of GROTESQUE test is {${\cal O}(1/poly(N)+1/poly(D))$}. {And in the setting $D={ \Theta}(N^{1-\delta})$ which is of primary interest, the error probability scales ${\cal O}(1/poly(D))$.}
\hfill \end{proof}

\noindent $b)$ {\it Decoding complexity:}\label{subsec:grotesquedecoding}


{Multiplicity testing only involves counting the total number of positives from $m_1$ tests, and hence the complexity is ${\cal O}(m_{1})$.}

{Localization testing involves decoding an expander code of block-length $m_2$, which is ${\cal O}(m_{2})$ by Theorem~\ref{thm:expandercode}.}

\section{Adaptive Group Testing}\label{sec:adaptive}
\begin{table*}[ht!]
\centering{}
\begin{tabular}{l|l}
\hline
\multicolumn{2}{l}{Adaptive Algorithm}\tabularnewline
\hline
${\cal S}$ & Left subset, ${\cal S}\subseteq{\cal N}$\tabularnewline
$\deg_{{\cal S}}(i)$ & The number of neighbors in left subsets\tabularnewline
${\cal D}$-zero nodes & $\{\mbox{right nodes \ensuremath{i}:}\deg_{{\cal D}}(i)=0\}$\tabularnewline
${\cal D}$-leaf nodes & $\{\mbox{right nodes \ensuremath{i}:}\deg_{{\cal D}}(i)=1\}$\tabularnewline
${\cal D}$-non-leaf nodes & $\{\mbox{right nodes \ensuremath{i}:}\deg_{{\cal D}}(i)\geq2\}$\tabularnewline
$T$ & The total number of stages\tabularnewline
$\Delta^{(t)}$ & The set of unrecovered defectives before the $t$-th stage tests are
performed, \tabularnewline
&$\Delta^{(1)}={\cal D}$ , $t\in\{1,\ldots,T\}$\tabularnewline
$d^{(t)}$ & The number of unrecovered defectives before the $t$-th stage tests
are performed, \tabularnewline
& $d^{(t)}=|\Delta^{(t)}|$, $t\in\{1,\ldots,T\}$\tabularnewline
${\cal G}^{(t)}$ & The bipartite graph for the $t$-th stage, $t\in\{1,\ldots,T\}$\tabularnewline
$r^{(t)}$ & The number of defectives recovered in $t$-th stage, $r^{(t)}=d^{(t)}-d^{(t+1)}$, $t\in\{1,\ldots,T-1\}$\tabularnewline
\hline
\end{tabular}
\caption{Table of notation used in our adaptive algorithm}
\end{table*}

In this section, we consider the {\em adaptive group testing} problem. The objective here is to determine an unknown set $\mathcal{D}$ of $D$ defective items from a collection $\mathcal{N}$ of size $N$. In this setting, we are allowed to perform tests sequentially in an adaptive manner, {\em i.e.}, the subset tested in each test may depend on the outcome of previous tests. As {stated earlier}, we assume that each test outcome may be incorrect independently with probability $q$.

\subsection{Overview}
\begin{figure*}[ht]
\centering\label{fig:adaptive}
\includegraphics[width=0.8\linewidth]{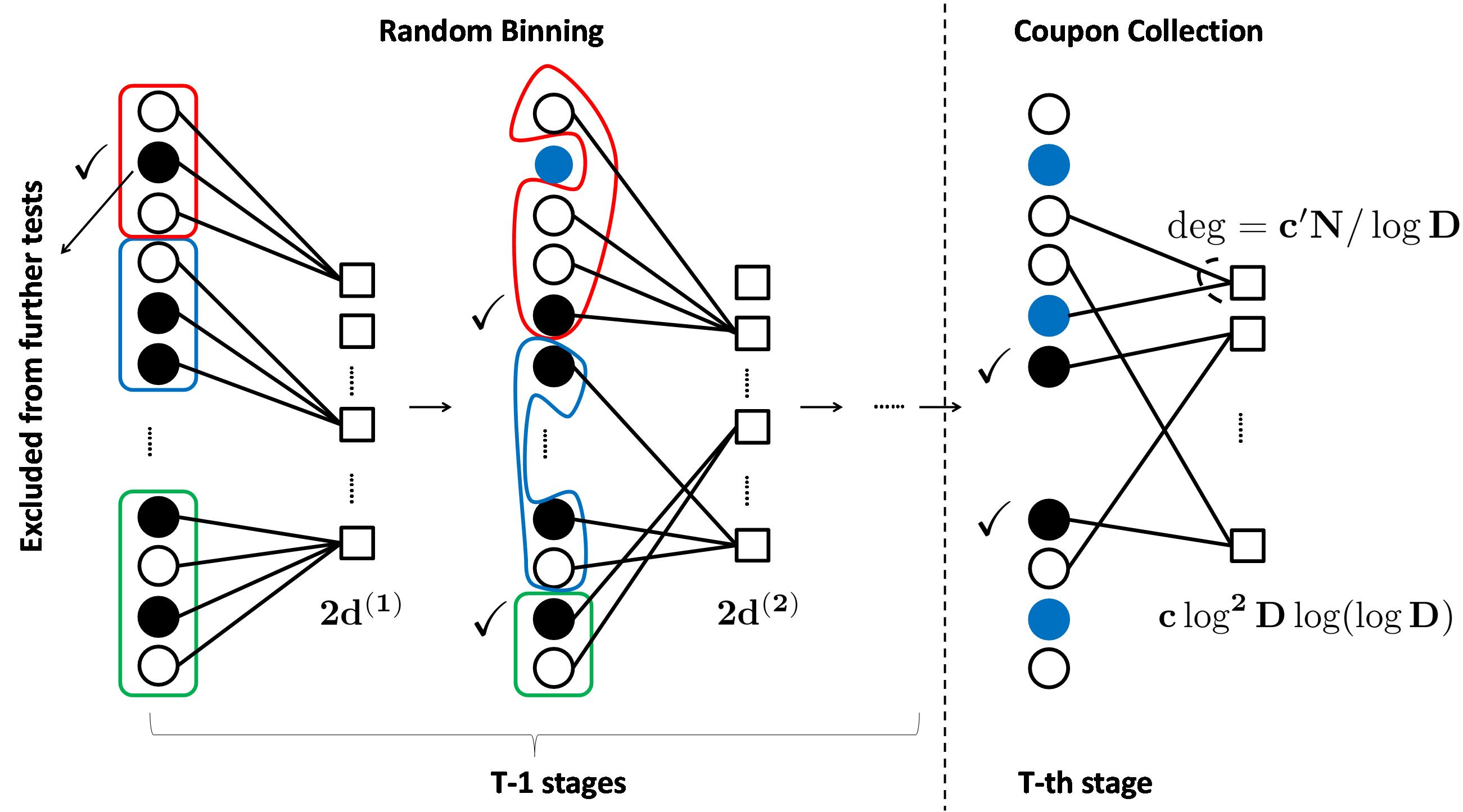}\caption{In each $t$-th stage ($t\leq T-1$), we generate a bipartite graph
with $N$ {nodes} on the left representing $N$ items and $2d^{(t)}$
nodes on the right. The black circular nodes represent defective items and the white
ones represent non-defective items. Each bipartite graph is left-regular with left-degree equal to $1$. The left nodes of each such graph are partitioned randomly -- different {coloured collections} show different ``pools'' within a partition. Nodes in the same pool connect
to the same right node. Each right node passes the items connected to it to the GROTESQUE tests. If there exists only one defective in one pool, then
GROTESQUE locates the defective item with high probability. For example, the second node in
the red partition of first graph will be detected by the GROTESQUE tests at the first right
node with high probability. In the next iteration, we exclude decoded defective items ({colored blue}) and
use a similar decoding process. Finally, in the $T$-th stage, we
generate a bipartite graph with $c(\log^{2}D)(\log(\log D))$ right nodes. For each right node, we pick its $c'N\log D$ neighbouring left nodes by choosing uniformly from all left nodes with replacement. By the coupon
collection argument, with high probability, we can decode the remaining
undecoded defective items. }
\end{figure*}

{Our algorithm has $\cO{\log (D)}$ adaptive stages.}
In all except the last stage, we design our tests so that with a high probability, in each stage, we recover a constant fraction of the remaining defectives. In each of these stages, the number of tests is roughly proportional to the number of remaining defectives. Thus, the total number of tests is $\cO{D\log{N}}$. In each of these stages, the tests are designed by first removing the defectives recovered so far, partitioning the set of remaining items into twice as many sets as the number of remaining defectives, performing GROTESQUE tests on each set from the partition. In our analysis, to guarantee that in each stage we recover a constant fraction of the remaining defectives with a high enough probability, we apply concentration inequalities. This works only when the number of unrecovered defectives be at least $\Omega(\log{D})$. Thus, we move on to the last stage when all but $\log{D}$ defectives have been recovered.

In the last stage, {we use the coupon collection problem~\cite{MotR:95} as a primitive, to identify the remaining defectives.} First we remove the set of defectives already recovered from the set of items being tested. Next, we design $\cO{\log{D}\log{\log{D}}}$ non-adaptive tests by picking random subsets of an appropriate size and performing GROTESQUE tests for that subset. We view the collection of outcomes from these sets of GROTESQUE tests as a random process that generates one independent ``coupon'' with constant probability each time. Thus, $\cO{\log{D}\log{\log{D}}}$ {such tests} suffice by standard coupon collector arguments.

Overall, our algorithm requires $\cO{D\log{N}}$ tests and runs in $\cO{D\log{N}}$ steps.

%

%

\subsection{Formal Description}\label{sec:formal-adaptive}
In this section, we describe {an} Adaptive Group Testing algorithm that achieves the following guarantees.
\begin{theorem} The Adaptive Group Testing algorithm described below has the following properties:

1) With probability {$1-{\cal O}(1/poly{(D)})$} over the choice of the random bipartite graph, the algorithm produces a reconstruction of the collection
$\hat{{\cal D}}$ of ${\cal D}$ such that $\hat{\mathbf{{\cal D}}}={\cal D}$.

2) The number of tests $M$ equals ${\cal O}(D\log N)$ .

3) The number of stages is $\cO{\log(D)}$.

4) The decoding complexity is ${\cal O}(D(\log(N))$.
\end{theorem}

Let $\Delta^{(t)}$, and $d^{(t)}$, respectively be the set, and the number, of unrecovered defectives before the $t$-th stage tests are performed. Note that $\mathcal{D}=\Delta^{(1)}\supseteq\Delta^{(2)}\ldots$ and $D=d^{(1)}\geq d^{(2)}\geq\dots$ {since we recover some defectives in each stage}. Let $T=1+\inf\{t:d^{(t)}\leq\log{D}\}$ denote the number of stages after which (with high probability) {the number of unrecovered defectives is no larger than $\log D$}.
{For any right node $i$, and subset of left nodes ${\cal S\subseteq{\cal N}}$ we define $\deg_{{\cal S}}(i)$ as the number of left nodes that neighbor the right node $i$. There are three types
of right nodes: {\it ${\cal D}$-leaf nodes}, {\it ${\cal D}$-zero nodes} and {\it ${\cal D}$-non-leaf nodes}. A ${\cal D}$-leaf
node is a right node $i$ with $\deg_{{\cal D}}(i)=1$ (it is connected to a single defective item), a ${\cal D}$-zero
node is a right node $i$ with $\deg_{{\cal D}}(i)=0$ (it is connected to no defective items), and a ${\cal D}$-non-leaf
node is a right node $i$ with $\deg_{{\cal D}}(i)\geq2$ (it is connected to multiple defective items).
Specifically, ${\cal D}$-leaf nodes are very helpful in our quick decoding process since
our GROTESQUE black-box shall with high probability correctly output the location of a defective item if there exists exactly one defective item among its neighbours.

\noindent{\bf Note:} {Even though in all our algorithms, both ${\cal D}$-zero nodes and $\cal D$-non-leaf nodes contain ``potentially useful'' information, we do not use them for decoding. This is because we require our algorithms to work with computational complexity that is comparable (up to constant factors for the adaptive algorithm, and at most a logarithmic factor in the other algorithms) to the size of the output of our algorithm (that is, $\cO{D\log(N)}$). To run this ``blindingly fast'', we need our algorithms to output ``something interesting'' in ``very little'' time. So, for instance, while the ${\cal D}$-zero nodes tell us which inputs are non-defective, using this information takes ``too long'' (since it essentially tells us which items are {\it not interesting} rather than those which {\it are} interesting, but there are many more non-interesting items than interesting ones).  }

\noindent $a)$ {\it Test design:}

{Conceptually, we design the first $T-1$ stages of our adaptive algorithm using the idea of ``random binning'', and the $T$-th stage using that of ``coupon collection''.}

\begin{itemize}

\item \noindent{\bf\underline{\em Random binning}:} For each stage $t=1,\ldots,T-1$, we consider a random left regular bipartite graph $\mathcal{G}^{(t)}$ with $2d^{(t)}$ right nodes and $N-(D-d^{(t)})$ left nodes corresponding to the set $(\mathcal{N}\setminus{D})\cup \Delta^{(t)}$, {\em i.e.,} all items except those already recovered in the previous $t-1$ stages. We let each left node of $\mathcal{G}^{(t)}$ be of degree one. We {choose} this graph uniformly at random from all possible bipartite graphs having mentioned properties above. Next, {for each right node of $\graph^{(t)}$, we use its set of left neighbours as the input for a GROTESQUE test. In each stage, for each right node of $\graph^{(t)}$, if GROTESQUE detects it has multiplicity one, it decodes the corresponding defective item. Else if GROTESQUE identifies a right-node as having multiplicity greater than one or zero, it does not further use these test outcomes.}

\item \noindent{\bf\underline{\em Coupon collection}:} In the final stage, we consider a left regular bipartite graph $\mathcal{G}^{(T)}$ with left node set $(\mathcal{N}\setminus\mathcal{D})\cup\Delta^{(T)}$ and $c(\log{D})^2\log{\log{D}}$ right nodes. For each right node, we choose its $c'N/\log{D}$ neighbours independently and uniformly at random with replacement. Next, we design $\cO{\log{N}}$ GROTESQUE tests at each right node to test for defectives among its $\cO{N}$ left neighbours.
\end{itemize}

\noindent $b)$ {\it Decoding algorithm:}\label{sec:adaptivealgo}

{The decoding for each stage corresponds to detection of leaf nodes in that stage and corresponding localization via GROTESQUE tests. Specifically, in each of the $t$-th stages for $t \leq T$, we sequentially pass the outputs of each of the right nodes of $\graph^{(t)}$ to GROTESQUE, which identifies leaf nodes and localized the corresponding defective items. Note that the structure of $\graph^{(T)}$ at the $T$-stage differs from the structure of $\graph^{(t)}$ for $t < T$, since they are chosen according to different processes (coupon collection versus random binning). Nonetheless, the same decoding procedure (leaf detection and corresponding localization of defectives) is performed in both the first $T-1$ stages, and th $T$th stage. The algorithm makes an error if not all defective items have been localized by the $T$-th stage (one can test for this by passing the set of remaining items to GROTESQUE's multiplicity testing subroutine).}

\subsection{Performance Analysis}
To analyze the performance of the first part of the algorithm, we require the following lemma.
\begin{lemma}\label{lem:mcdiarmid}
Let $\mathcal{G}$ be a random bipartite graph with $n$ nodes on the left, and $2d$ nodes on the right side. Let $\Delta$ be a subset of the left nodes of size $d$. Also, let each node on the left side of $\mathcal{G}$ have degree one. Then, for any $\epsilon>0$, with probability at least $1-\exp(-\epsilon^2d/2)$, {at least} $(e^{-1/2}-\epsilon)d$ nodes on the left are connected to $\Delta$-leaf nodes.
\end{lemma}
\begin{proof}
We first note that the probability that a node $j\in\Delta$ on the left {of the bipartite graph is} connected to a $\Delta$-leaf {node} is:
\begin{eqnarray*}
\Pr(j\mbox{ is connected to a $\Delta$-leaf})&=&\left(\frac{2d-1}{2d}\right)^{d-1}\\
&>&\left(1-\frac{1}{2d}\right)^{(2d-1)/2}>e^{-1/2},\\
\end{eqnarray*}
where the last inequality comes from the well-known inequality: $\left(1-\frac{1}{x}\right)^{x-1}>e^{-1}$.
{Therefore the expected number of nodes from $\Delta$ that are connected to $\Delta$-leaf nodes is $e^{-1/2}d$.} Next, we show a concentration result for this number by applying McDiarmid's inequality~\cite{Mcd:89} the statement of which is reprised in Appendix A as follows. First, we label nodes on the right with numbers {from} $1$ to $2d$ and for each {$i\in 1,2,\ldots,d$}, let
\[
{\mathbf{X_{i}}}\triangleq\mbox{label of the node on the right which is connected to the \ensuremath{i}-th node in $\Delta$.}
\]

Also, define ${f:\mathbf{X_1}\times \mathbf{X_2} \times \cdots \times \mathbf{X_d} }\rightarrow \mathbb{Z}$ as the number of left nodes connected to a $\Delta$-leaf node {\it i.e.},
\begin{eqnarray}
f(x_1,x_2,\cdots,x_d)&\triangleq &|\{x_1,x_2,\cdots,x_d\}|
\end{eqnarray}

It is observed that for any fixed $x_1,\cdots,x_{i-1},x_{i+1},\cdots,x_d$ and any $x_i,x'_i\in \mathbf{X_i}$,
\begin{eqnarray}
|f(x_1,\cdots,x_i,\cdots,x_d)-f(x_1,\cdots,x'_i,\cdots,x_d)|\leq 2\nonumber
\end{eqnarray}

For example, $\exists i,j(\neq i)$, $x_{i}=x_{j}$ and $x_{i}'\neq x_{i}$. It's possible
that $x_{j}$-th and $x_{i}'$-th right nodes which initially are
not $\Delta$-leaf nodes become $\Delta$-leaf nodes. Then, $f(x_1,\cdots,x_i,\cdots,x_d)-f(x_1,\cdots,x'_i,\cdots,x_d)= -2$. In fact, we can numerate all possible cases to conclude out result.

It is clear that $X_i$'s are independent. Therefore, by McDiarmid's inequality we have:
\begin{eqnarray}
\Pr\left(f(X_1,X_2,\cdots,X_d)-{\bf E}f(X_1,X_2,\cdots,X_d)<-\epsilon d\right)\leq\exp\left(-\epsilon^2d/2\right).\nonumber
\end{eqnarray}
Hence,
\begin{eqnarray*}
\Pr\left(f(X_{1},X_{2},\ldots,X_{d})-e^{-1/2}d<-\epsilon d\right) & \leq & \Pr\left(f(X_{1},X_{2},\ldots,X_{d})-{\bf E}f(X_{1},X_{2},\ldots,X_{d})<-\epsilon d\right)\\
 & \leq & \exp\left(-\epsilon^{2}d/2\right)
\end{eqnarray*}
\hfill \end{proof}
\begin{corollary}\label{cor:adaptive} With probability $1-\cO{\exp{(-D)}+(\log{D})/N}$, in each of the first $T-1$ stages, our decoding algorithm recovers at least a $e^{-1/2}/2$ fraction of defectives that have not been decoded up to that stage.
\end{corollary}
\begin{proof}The event that we recover fewer than $e^{-1/2}/2$ fraction of defectives in $t$-th stage is a subset of the union of the following two events:

\begin{enumerate}[(1)]
\item Less than a $e^{-1/2}/2$ fraction of defectives are connected to a $\Delta^{(t)}$-leaf node.
\item The outcome of the GROTESQUE tests is incorrect for any of the $2d^{(t)}$ right nodes.
\end{enumerate}

By Lemma~\ref{lem:mcdiarmid}, with $\Delta=\Delta^{(t)}$, $n=N-d^{(t)}$, and $\epsilon=e^{-1/2}/2$, the probability of event~(1) is at most $\exp(-d^{(t)}/4e)$. Further, by Lemma~\ref{lem:errorgrotesque}, with $n=N-d^{(t)}$, the probability of event~B is $\cO{1/N}$. For each $t$, let $r^{(t)}=d^{(t)}-d^{(t+1)}$. Therefore, the probability that in one of the $T-1$ stages, fewer than a $e^{-1/2}/2$ fraction of defectives are correctly decoded is bounded from above as

\[
\Pr\left(\cup_{t=1}^{T-1}\{r^{(t)}<e^{-1/2}d^{(t)}/2\} \right)= \sum_{t=1}^{T-1}\Pr\left(r^{(t)}<\frac{e^{-\frac{1}{2}}d^{(t)}}{2}\left\vert\bigcap_{\tau=1}^{t-1}\left\{r^{(\tau)}>\frac{e^{-\frac{1}{2}}d^{(\tau)}}{2}\right\}\right.\right).
\]
Under the conditioning event for the $t$-th term, we may bound $d^{(t)}$ from above by $D(1-e^{-1/2}/2)^{t-1}$. Further, by the definition of $T$, there are at most $\log{D}/\log{(1-e^{-1/2}/2)}-1$ terms. The chain of inequalities is further simplified as
\begin{eqnarray}
\lefteqn{\Pr\left(\cup_{t=1}^{T-1}\{r^{(t)}<e^{-1/2}d^{(t)}/2\} \right)}\nonumber\\
&\leq  &\cO{\sum_{t=1}^{{\log{D}}/{\log{(1-\frac{e^{-\frac{1}{2}}}{2})}}-1}\exp{(-D(1-e^{-1/2}/2)^{t-1})}+1/N}\nonumber\\
&=&\cO{\exp{(-D)}+(\log{D})/N}\nonumber
\end{eqnarray}
\hfill \end{proof}

\noindent $a)$ {\it Number of tests:}

In the Random Binning part of the algorithm, there are $2d^{(t)}$ right nodes in the $t$-th stage, each of which  requires $C{d^{(t)}\log{N}}$ tests for some constant $C=C(q)$ (as determined in Section~\ref{sec:grotesque}). Hence the total number of tests for the Random Binning part of the algorithm is $\sum_{t=1}^{T-1} 2Cd^{(t)}\log{N}$.  By Corollary~\ref{cor:adaptive}, with high probability, in each stage the algorithms recovers a constant fraction of the undecoded defectives. Thus, $d^{(t)}=\cO{\alpha^t D}$ for some constant $\alpha\in(0,1)$. Therefore, the total number of tests required in the first $T-1$ stages is $\cO{D\log{N}}$.

For the Coupon Collection stage, the number of right nodes is $\cO{\log{D}\log{\log{D}}}$. Since we perform GROTESQUE tests for each right node, the total number of tests required is $\cO{\log{D}\log{\log{D}}\log{N}}$, which is less than the $\cO{D\log{N}}$ tests required in the Random Binning part of the algorithm.

Thus our proposed scheme requires $\cO{D\log{N}}$ tests overall.

\noindent $b)$ {\it Decoding complexity:}

Since in each stage our decoding algorithm has to step through all right nodes to decode the corresponding GROTESQUE tests, the total number of right nodes the algorithm needs to consider is $\sum_{t=1}^{T-1}d^{(t)}+\cO{\log{D}\log{\log{D}}}=\cO{D}$. By
 the analysis from Section~\ref{subsec:grotesquedecoding}, decoding each collection of GROTESQUE tests at a right node takes $\cO{\log{N}}$ time. Therefore, our algorithm runs in $\cO{D\log{N}}$ time.

\noindent $c)$ {\it Error probability:}

We show that the above algorithm succeeds with high probability in decoding all the defectives  in the claimed number of stages.

By the analysis of error events (1) and (2) in Corollary~\ref{cor:adaptive}, in the first part of the algorithm (random binning) we recover at least $D-\log{D}$ defectives with probability $1-\cO{\exp(-D)+(\log{D})/N}$.

To analyze the last (coupon collection) stage of tests, note that an error may occur only if one of the following events occur,
\begin{enumerate}[(1)]
\item[(3)] In the coupon collection process, fewer than $\log{D}$ distinct coupons are collected.
\item[(4)] At least one of the collected coupons was incorrectly decoded.
\end{enumerate}
To bound the probability of event~(3), we apply standard concentration bounds on the coupon collector's problem~\cite{MotR:95}. Towards this end, we first note that the probability that our algorithm identifies a right node $i$ as a leaf node may be written as
\begin{eqnarray}
\lefteqn{\Pr(\mbox{$i$ decoded as a $\Delta^{(T)}$-leaf node})}\nonumber\\
&\geq &\Pr(\mbox{$i$ decoded as a $\Delta^{(T)}$-leaf node, $i$ is a $\Delta^{(T)}$-leaf node})\nonumber\\
&=&\Pr(\mbox{$i$ decoded as a $\Delta^{(T)}$-leaf node $|$ $i$ is a $\Delta^{(T)}$-leaf node})\times\Pr(\mbox{$i$ is a $\Delta^{(T)}$-leaf node})\nonumber\\
&=&(1-\cO{1/N})\times\Pr(\mbox{$i$ is a $\Delta^{(T)}$-leaf node})\nonumber\\
&= &\Theta(1)\times \frac{cN}{\log{D}}\frac{\log D}{N-D+\log{D}}\left(1-\frac{\log{D}}{N-D+\log{D}}\right)^{\frac{cN}{\log{D}}}\nonumber\\
&= &\Theta(1).\nonumber
\end{eqnarray}

Since each decoded leaf node is independently chosen and the probability of picking a coupon is constant, by tail bounds on the coupon collection process~\cite{MotR:95}, the probability that at least one coupon has not been collected in $c'(\log{D})^2\log\log D$ steps is $\cO{(\log D)^{-\log{D}}}=\cO{poly(1/D)}$.

Thus, the overall error probability decays as $\cO{poly(1/D)}$.

\section{Non-adaptive Algorithms}
\label{sec:nonadaptive}
\begin{table*}[ht!]
\centering{}

\begin{tabular}{l|l}
\hline
\multicolumn{2}{l}{Non-adaptive Algorithm}\tabularnewline
\hline
${\cal G}_{g}$ & $g$-th sub bipartite graph, $g\in\{1,\ldots,c_{1}\log D\}$\tabularnewline
$\hat{{\bf y}}_{i}$ & The length-${\cal O}(\log N)$ binary vector denoting the actually
observed noisy outcomes \tabularnewline
 & of the $i$-th GROTESQUE Encoding\tabularnewline
$c_{1}$ & A constant related to the number of sub bipartite graph\tabularnewline
$c_{2}$ & A constant related to the number of right nodes for each bipartite graph\tabularnewline
$c_{3}$ & A constant related to the number of multiplicity tests for each right node\tabularnewline
$c_{4}$ & A constant related to the number of localization tests for each right node\tabularnewline
$c_{5}$ & $c_{5}=c_{1}c_{2}$\tabularnewline
$P_{0}$ & The probability that the neighbor of a defective item is a $\cal D$-leaf node\tabularnewline
${\cal L}({\cal D})$ & Leaf node list which contains all the ${\cal D}$-leaf nodes\tabularnewline
${\cal L}^{(t)}$ & Leaf node list for the $t$-th iteration.\tabularnewline
$M_{i,1}$ & The number of Multiplicity tests for right node $i$\tabularnewline
$M_{i,2}$ & The number of Localization tests for right node $i$\tabularnewline
\hline
\end{tabular}
\caption{Table of notation used in our Non-adaptive algorithm}

\end{table*}

We consider non-adaptive group testing in this section. In non-adaptive
group testing, the set of items being tested in each test is required
to be independent of the outcome of every other test~\cite{DuH:00}.

The objective here is to determine an unknown set ${\cal D}$ of $D$
defective items from a collection ${\cal N}$ of size $N$. We assume
that each test outcome may be incorrect independently with probability
$q$.

\subsection{Overview}\label{sec:overview_nonadaptive}
\begin{figure}
\centering
\includegraphics[width=0.6\linewidth]{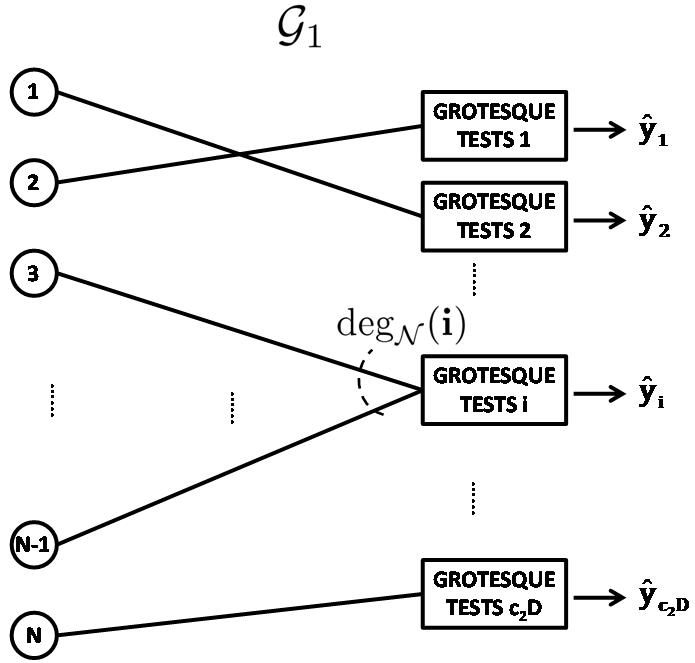}\caption{We generate $c_{1}\log D$ bipartite graphs with $N$ nodes on the left (representing $N$ items) and $c_{2}D$ nodes on the right. Take ${\cal G}_{1}$ as an example. For each right node $i$, we generate ${\cal O}(\log N)$ tests, $\bf{\hat{y_{i}}}$, by GROTESQUE tests. The input of the $i$-th GROTESQUE TESTS are the items connected to the right node $i$. The size of outputs is ${\cal O}(\log N)$. Based on the properties of GROTESQUE TESTS, we can estimate whether there exists exactly one defective item and if so, we can estimate its location with high probability.}\label{fig:nonadaptive1}
\end{figure}

\begin{figure}[t]
\centering\includegraphics[width=0.8\linewidth]{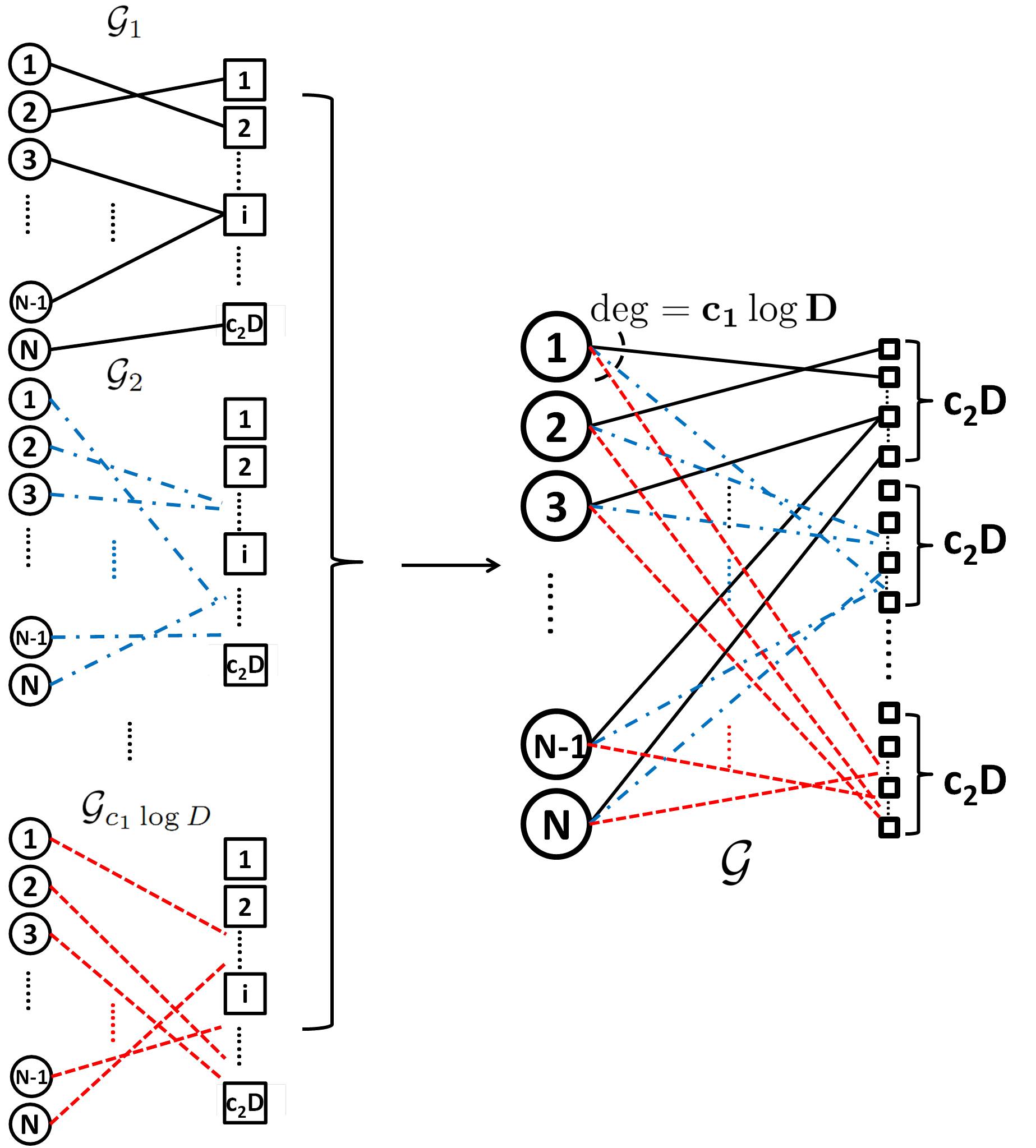}\caption{After generating $c_{1}\log D$ bipartite graphs as in Figure~\ref{fig:nonadaptive1}, we combine them to obtain the overall graph ${\cal G}$ for our non-adaptive group testing algorithm.
There are $N$ nodes on the left representing $N$ items. Collect the right nodes of all the $c_{1}\log D$ bipartite graphs on the right
side of ${\cal G}$ and maintain the connectivity between two sides. Different
colors show the connectivity between $N$ items and different right
node sets of equal size $c_{2}D$. Finally, we get a bipartite graph
${\cal G}$ with left regularity equal to $c_{1}\log D$. We can guarantee
that, with high probability, each defective item connects to a ${\cal D}$-leaf
node. }\label{fig:nonadaptive2}

\end{figure}

The structure of group-testing tests is based on left-regular bipartite
graphs $\{{\cal G}_g\}$. We put $N$ items on the left-hand side and $c_{5}D\log D$
nodes on the right-hand side of a bipartite graph. Each node on the
right-hand side of a bipartite graph is called a {\em group-testing} node. Group tests corresponding to the multiplicity and localization tests of GROTESQUE are performed as part of the encoding process, but the results are not necessarily used to decode. This is because our decoding algorithm can cherry-pick the group-testing nodes to decode only the ``useful'' ones.

The set of ${\cal D}$-leaf nodes (see the definition in Section \ref{sec:formal-adaptive}) are helpful for our decoding process since
GROTESQUE tests performed on a right node output the location of a defective item only if there
exists exactly one defective item among its neighbours. In our iterative
algorithm, we claim that there exists at least one ${\cal D}$-leaf
node in each iteration, so that we can decode one defective item.

\subsection{Formal Description}
In this section, we describe a probabilistic construction of the tests and an iterative Non-Adaptive Group Testing algorithm
that achieves the following guarantees.
\begin{theorem} The Non-Adaptive Group Testing algorithm described below has the following properties:

1) With probability {$1-{\cal O}(1/poly{(D)})$} over the choice of the random bipartite graphs, the algorithm produces a reconstruction of the collection
$\hat{{\cal D}}$ of ${\cal D}$ such that $\hat{\mathbf{{\cal D}}}={\cal D}$.

2) The number of tests $M$ equals ${\cal O}(D\log D\log N)$.

3) The decoding complexity is ${\cal O}(D(\log(N)+\mbox{\ensuremath{\log}}^{2}D))$.
\end{theorem}

\noindent $a)$ {\it Description of graph properties:}

We first construct a bipartite graph ${\cal G}$ (Figure~\ref{fig:nonadaptive2}) with some
desirable properties outlined below. We then show that such the random graphs we choose satisfy such properties with high probability. In Section V-Bb), we then use these graph properties in
the Non-adaptive Group Testing algorithm.

\noindent $\underline{{\bf Properties}\mbox{ }{\bf of}\mbox{ }{\bf {\cal G}:}}$
\begin{enumerate}
\item \underline{\em Construction of a left-regular bipartite graph:}
As in Figure~\ref{fig:nonadaptive2}, we generate ${\cal G}$ by combining $c_{1}\log D$ left-regular graphs
${\cal G}_{g}$, for $g=1,\ldots,c_{1}\log D$. For each ${\cal G}_{g}$,
there are $N$ items on the left and $c_{2}D$ group-testing nodes
on the right. The graph ${\cal G}_{g}$ has left-regularity equal to $1$ and
each edge connects to a right node uniformly at random. After constructing
all the $c_{1}\log D$ graphs ${\cal G}_{g}$, we combine them to form ${\cal G}$
in the following way. Keep $N$ items on the left, and collect the
right nodes. Therefore, ${\cal G}$ has the properties that it has
$N$ nodes on the left with left-regularity equal to $c_{2}\log D$, and $c_{5}D\log D=c_{1}c_{2}D\log D$ nodes on the right.
\item \underline{\em ``Many" ${\cal D}$-leaf nodes:} For
any set ${\cal D}$ of size $D$ on the left of ${\cal G}$, none of the nodes in ${\cal D}$ has fewer than a
constant fraction of ${\cal D}$-leaf nodes. The proof of this statement is the subject of Lemma
\ref{lem:many leaf}.
\end{enumerate}
\begin{lemma}
\label{lem:many leaf}With probability $1-{\cal O}(D^{-1})$, the
fraction of $c_{1}\log D$ neighbors of each defective item that are ${\cal D}$-leaf
nodes is $\mbox{\ensuremath{\Omega}(}\exp(-1/c_{2}))$ .
\end{lemma}
\begin{proof}
Define $\mathbf{W_{i,j}}$ as the random variable representing whether the
neighbor of a defective item $x_{j}$ is a ${\cal D}$-leaf node for
the $i$-th graph ${\cal G}_{i}$.
\[
\mathbf{W_{i,j}}=\left\{ \begin{array}{cc}
1, & \mbox{if the neighbor left node \ensuremath{j} is a \ensuremath{{\cal D}}-leaf node}\\
0, & \mbox{otherwise}
\end{array}\right.
\]
Then, the total number of ${\cal D}$-leaf nodes of $x_{j}$ is $\sum_{i=1}^{\log D}\mathbf{W_{i,j}}$.

\begin{eqnarray*}
P_{0} & \triangleq & \Pr(\mbox{the neighbor of a defective }\\
 &  & \mbox{item \ensuremath{x_{j}} is a \ensuremath{{\cal D}}-leaf node})\\
 & = & \left(1-\frac{1}{c_{2}D}\right)^{D-1}\\
 & \approx & \exp\left(-\frac{1}{c_{2}}\right),
\end{eqnarray*}
\noindent and $\mathbf{W_{1,j}},\ldots,\mathbf{W_{c_{1}\log D,j}}$ are i.i.d. Bernoulli random
variables with parameter $P_{0}$.

Therefore, by the Chernoff bound, we have
\begin{eqnarray*}
 &  & \Pr\left(\sum_{i=1}^{c_{1}\log D}\mathbf{W_{i,j}}-P_{0}c_{1}\log D\leq-\epsilon P_{0}c_{1}\log D\right)\\
 & \leq & \exp\left(-\frac{\epsilon^{2}}{2}P_{0}c_{1}\log D\right)
\end{eqnarray*}

Hence, the probability for each defective item that it has at least a
constant fraction of ${\cal D}$-leaf nodes is $1-{\cal O}(1/D)$
(by choosing the $c_{1}$ and $c_{2}$ appropriately).

Then, by the union bound, all defective items have at least constant fraction
of ${\cal D}$-leaf nodes.
\hfill \end{proof}

In the following two sections, we describe how to use the properties
of ${\cal G}$ to perform the encoding and decoding.

\noindent $b)$ {\it Test design:}

For each node $i$ on the right of ${\cal G}$, we design GROTESQUE tests
with $\deg_{{\cal N}}(i)$ inputs which are the items connected to right node $i$. We choose the
number of multiplicity tests, $M_{i,1}$, equal to $c_{3}\log(c_{5}D\log D)$
and the number of localization tests, $M_{i,2}$, equal to $c_{4}\log(\deg_{{\cal N}}(i))$
for some positive constant $c_{3}$ and $c_{4}$. We already know
that $c_{3}\log(c_{5}D\log D)={\cal O}(\log D)$. Since the right degree of any node is at most $N$, the total number of GROTESQUE-tests
required for the right node $i$ equal to $M_{i,1}+M_{i,2}={\cal O}(\log N)$.
Therefore, the overall number of tests $M$ equals ${\cal O}(D\log D\log N)$.

\noindent $c)$ {\it Decoding algorithm:}

Before the iterative decoding process, we make a ${\it leaf}\mbox{ }node\mbox{ }list$,
${\cal L}({\cal D})$, which contains all the ${\cal D}$-leaf nodes
based on the multiplicity testing part of GROTESQUE tests. Based on the
properties of graph ${\cal G}$, we know that each defective item has at least a constant fraction of ${\cal D}$-leaf nodes.
Denote ${\cal L}^{(t)}$ as the leaf node list in $t$-th iteration,
$t=1,2,\ldots,D,D+1$. ${\cal L}^{(1)}={\cal L}({\cal D})$, ${\cal L}^{(D+1)}=\varnothing$
and ${\cal L}^{(t)}\neq\varnothing$, for $t=1,2,\ldots,D$. In the
$t$-th iteration, we pick a right node $i\in{\cal L}^{(t)}$ and
decode a defective item using the localization part of GROTESQUE tests
of $i$ to locate the corresponding defective item. After that, we
cancel the defective item, its corresponding edges and its neighbors.
We update the ${\cal L}^{(t)}$ to ${\cal L}^{(t+1)}$. In the $(t+1)$-th
iteration, we pick another ${\cal D}$-leaf node in ${\cal L}^{(t+1)}$.
The formal description of the non-adaptive group testing algorithm is as follows:
\begin{enumerate}
\item $\underline{Initialization:}$ Go through all the right nodes and
use {\em only} the multiplicity testing part of GROTESQUE-test to initialize ${\cal L}^{(1)}={\cal L}({\cal D})$.
\item $\underline{Operations\mbox{ }in\mbox{ }t\mbox{-th }iteration:}$

\begin{enumerate}[i)]
\item \label{enu:Pick-a-right-node}Pick any right node $j$ in ${\cal L}^{(t)}$;
\item Use localization part of GROTESQUE-test to decode the corresponding defective;
\item Remove the decoded defective item, all the edges connected to it, and
all its neighbours;
\item Update ${\cal L}^{(t)}$ to ${\cal L}^{(t+1)}$ by removing the leaf nodes removed in step iii) and return to step i).
\end{enumerate}
\item $\underline{Termination:}$ The algorithm stops when the leaf node
list becomes empty, and outputs the defective set $\hat{{\cal D}}$.
\end{enumerate}

\subsection{Performance Analysis}

\noindent $a)$ {\it Number of iterations:}

Since we guarantee that in each iteration we can decode one defective item, the number of iterations is $D$.

\noindent $b)$ {\it Decoding complexity:}


For each right node, checking the multiplicity testing part of GROTESQUE-test
costs $c_{3}\log(c_{5}D\log D)$ steps. Therefore, the total computational cost is $[c_{3}\log(c_{5}D\log D)](c_{5}D\log D)={\cal O}(D\log^{2}D)$ in the initialization step.

In each iteration, the cost of localization is ${\cal O}(\log N)$ steps. Removing the decoded defective item takes $1$ step. Removing all the edges
connected to it and its neighbours takes time $\log(c_{5}D\log D)](c_{1}\log D)={\cal O}(\log^{2}D)$,
where $\log(c_{5}D\log D)$ is the cost of addressing one neighbour of
the decoded defective item. Therefore, the time complexity for the
iterative decoding process is ${\cal O}(D(\log N+\log^{2}D))$.

Hence, we can conclude that the overall time complexity is ${\cal O}(D(\log N+\log^{2}D))$
based on the analysis above.

\noindent $c)$ {\it Error probability:}


Finally, we show that $\hat{{\cal D}}={\cal D}$ with a high probability
by choosing the parameters $c_{3}$ and $c_{4}$ carefully.

We set $M_{i,1}=c_{3}\log(c_{5}D\log D)$ to make sure that the probability
of incorrect multiplicity decoding for each node is ${\cal O}(1/c_{5}D\log D)$
by choosing $c_{3}>1$. Then by the union bound, the probability of incorrect
multiplicity decoding is $\cO{1/poly(D)}$.

We set $M_{i,2}=c_{4}\log(\deg_{{\cal N}}(i))={\cal O}(\log N)$
to make sure that each neighbour of $i$ has distinct codewords by choosing
$c_{4}>1$. The probability of incorrect localization in each iteration is ${\cal O}(\exp(-M_{i,2}))$ which is upper bounded by
${\cal O}(1/N)$. Finally, by applying the union bound over $D$ iteration,
the probability of incorrect decoding is bounded from above by $\cO{1/poly(D)}$.

\section{Two-stage group testing}\label{sec:2stage}
In this section, we present a $2$-stage adaptive group testing problem with both  decoding complexity and number of tests that is nearly order-optimal (up to a multiplicative factor that is at most $\cO{\log N}$). Again, the objective is to determine an unknown
set ${\cal D}$ of $D$ defective items from a collection ${\cal N}$
of size $N$. In both stages, we perform tests in a non-adaptive manner,
though the tests of the second stage depend on the outcomes of tests
in the first stage. As earlier, we assume that each test outcome
may be incorrect independently with probability $q$.

\subsection{Overview}
\begin{figure*}[ht]
\centering\label{fig:2-stage}
\includegraphics[width=0.7\linewidth]{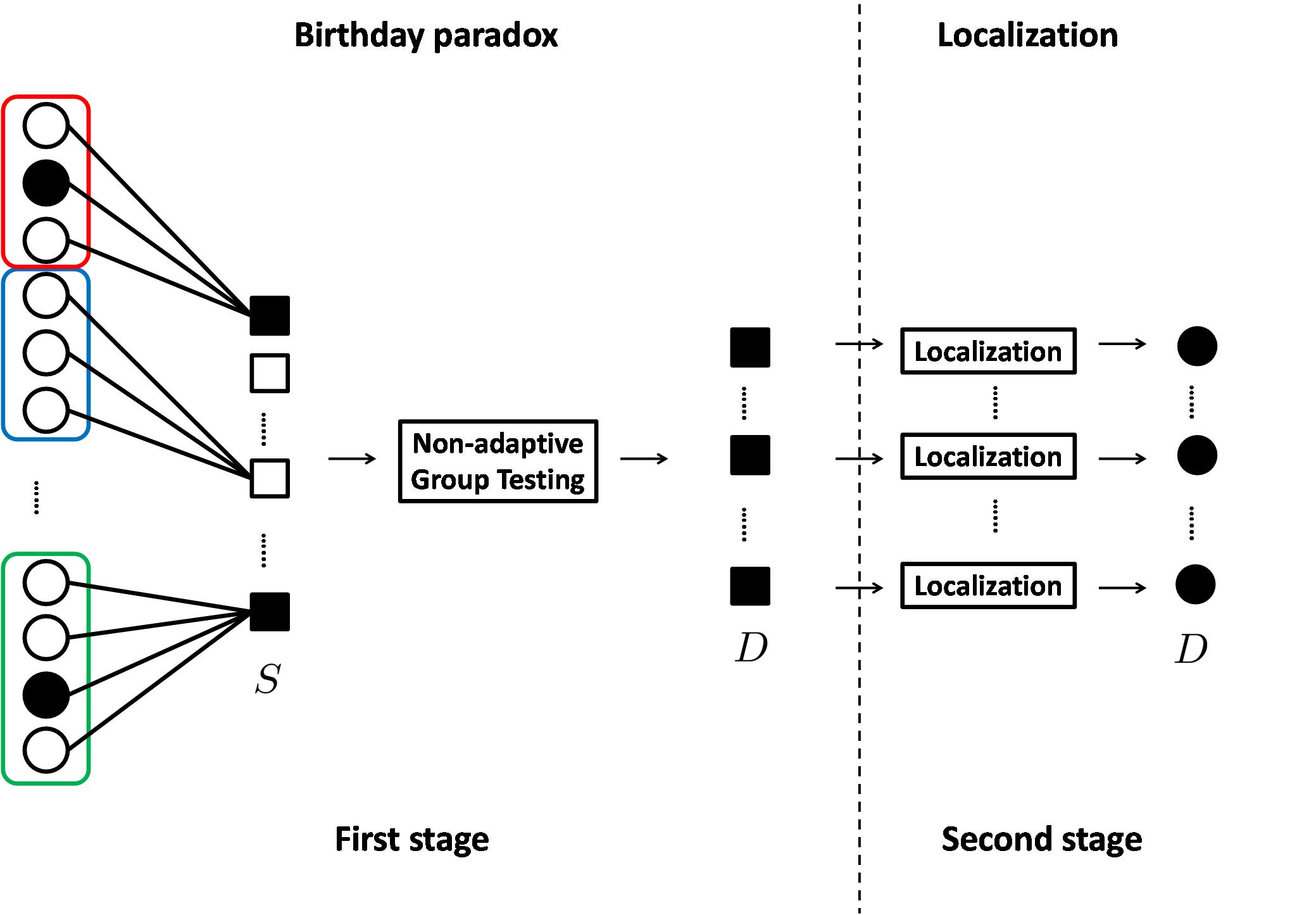}\caption{In the first stage, we generate a bipartite graph with $N$ nodes on the left representing $N$ items and $S$ nodes on the right. The black circular nodes represent defective items and the white ones represent non-defective items. Each bipartite graph is left-regular with left-degree equal to $1$. The left nodes of such graph are partitioned randomly -- different coloured collections have different ``birthdays''. Nodes in the same partition have the same ``birthday''. With high probability, each right node is either ${\cal D}$-leaf node (black right node) or ${\cal D}$-zero node (white right node) according to our choice of $S$. Applying our non-adaptive algorithm, we identify $D$ leaf nodes. In the second stage, applying localization testing on each ${\cal D}$-leaf node, we identify the corresponding defective items.}
\end{figure*}

\begin{table*}[ht!]
\centering{}

\begin{tabular}{l|l}
\hline
\multicolumn{2}{l}{Two-Stage Adaptive Algorithm}\tabularnewline
\hline
${\cal G}$ & A bipartite graph used in the first stage\tabularnewline
$S$ & The number of nodes on the right of ${\cal G}$\tabularnewline
\hline
\end{tabular}
\caption{Table of notation used in our 2-stage adaptive algorithm}

\end{table*}

Our algorithm has $2$ adaptive stages.

In the first stage, we use the birthday paradox problem as a primitive
to construct a bipartite graph ${\cal G}$. ${\cal G}$ has the following
properties - the graph is bipartite, has $N$ nodes on the left representing
$N$ items ($N$ ``people''), is left-regular with regularity equals $1$, and $S$
($=poly(D)$ with degree larger than $3$) nodes on the right ($S$ choices of ``birthdays''). We show that with
high probability ($1-\cO{1/poly(D)}$), each right node is either a ${\cal D}$-leaf node
or a ${\cal D}$-zero node ({\it i.e.}, no pair of them have the same birthday). We use the non-adaptive algorithm discussed
in Section \ref{sec:nonadaptive} on the $S$ right nodes to identify the $D$ ${\cal D}$-leaf
nodes. In the first stage the total number of tests is ${\cal O}(D\log D\log S)={\cal O}(D\log^{2}D)$
and the decoding complexity is ${\cal O}(D(\log S+\log^{2}D))={\cal O}(D\log^{2}D)$.

In the second stage, we use the localization procedure of  GROTESQUE with $n = \cO{(N/poly(D))}$, on
all ${\cal D}$-leaf nodes identified in the first stage.
Note that with high probability there are exactly $D$ right nodes that are ${\cal D}$-leaf nodes, out of $poly(D)$ right nodes in total -- the fact that we test only $D$ of them is what gives us potentially significant savings in the number of tests and decoding complexity.
In the second stage the total number of tests is ${\cal O}(D\log N)$
and the decoding complexity is ${\cal O}(D\log N)$.

Over both stages, our $2$-stage adaptive algorithm hence requires ${\cal O}(D(\log N+\log^{2}D))$
tests and runs in ${\cal O}(D(\log N+\log^{2}D))$ steps.

\subsection{Formal Description}

In this section, we describe a $2$-stage adaptive group testing algorithm
that achieves the following guarantees.
\begin{theorem}
The Two-stage Adaptive Group Testing algorithm described below has
the following properties:

1) With probability $1-{\cal O}(1/poly(D))$ over the choice of the random bipartite
graph, the algorithm produces a reconstruction of the collection $\hat{{\cal D}}$
of ${\cal D}$ such that $\hat{{\cal D}}={\cal D}$.

2) The number of tests $M$ equals ${\cal O}(D(\log N+\log^{2}D))$.

3) The number of stages is $2$.

4) The decoding complexity is ${\cal O}(D(\log N+\log^{2}D))$.
\end{theorem}

\noindent $a)$ {\it Test design and decoding algorithm:}

\begin{itemize}
\item \noindent{\bf\underline{\em Birthday paradox hashing}:} In the first stage, we consider
a random left regular bipartite graph ${\cal G}$ with $S$ right nodes
and $N$ left nodes. We set each left node of ${\cal G}$ to be of degree
one. We choose this graph uniformly at random. The property we required
is that, with high probability, each right node is either a ${\cal D}$-leaf
node or a ${\cal D}$-zero node (see the definitions in Section~\ref{sec:formal-adaptive}).
By the ``standard birthday paradox argument'' \cite{feller1}, the failure probability
scales as ${\cal O}(1/poly(D))$ if we choose $S={\cal O}(poly(D))$
with degree larger than $3$ (see Lemma~\ref{birthday} below). To identify all the
${\cal D}$-leaf nodes is equivalent to the group testing problem of finding $D$
defectives from $S$ items. We apply our non-adaptive algorithm to
all right nodes. Here if a right node $i$ is (respectively is not)
included in a test, then all the neighbors of $i$ are (respectively
are not) included in that test. The outcomes of the first stage are
all ${\cal D}$-leaf nodes.

\item \noindent{\bf\underline{\em Localization}:} In the second stage, we use the GROTESQUE localization procedure (with $n = \cO{N/poly(D)}$) on
each ${\cal D}$-leaf node identified in the first stage, to decode the
corresponding defective item.
\end{itemize}

\subsection{Performance Analysis}

The analyze the performance of the first part of the algorithm, we
require the following lemma.
\begin{lemma}\label{birthday}
The probability that no defective items have the same neighbor (right
node) scales as $1-{\cal O}(1/poly(D))$, if we choose $S={\cal O}(poly(D))$
with degree larger than $3$.
\end{lemma}
\begin{proof}
There are at least two ways to prove the correctness of this lemma.

First, using an argument similar to that in the birthday paradox problem,
\begin{eqnarray*}
\lefteqn{\Pr(\mbox{No defective items have the same neighbor})}\\
 & = & \frac{{S \choose D}D!}{S^{D}}\\
 & = & \left(1-\frac{1}{S}\right)\left(1-\frac{2}{S}\right)\ldots\left(1-\frac{D-1}{S}\right)\\
 & = & \prod_{i=1}^{D-1}\left(1-\frac{i}{S}\right)
\end{eqnarray*}

Using $1-\frac{i}{S}\approx e^{-i/S}$ when $i\ll S$,

\begin{eqnarray*}
\prod_{i=1}^{D-1}\left(1-\frac{i}{S}\right) & \approx & \prod_{i=1}^{D-1}e^{-i/S}\\
 & = & e^{-\sum_{i=1}^{D-1}i/S}\\
 & = & e^{-D(D-1)/2S}\\
 & \approx & e^{-D^{2}/2S}\\
 & \approx & 1-\frac{D^{2}}{2S}
\end{eqnarray*}

Therefore, if we choose $S={\cal O}(poly(D))$ with degree larger
than $3$, the probability that each right node has no more than two
defective items is $1-{\cal O}(1/poly(D))$.

For an alternative proof, we consider the probability that the event considered in the statement of this lemma does not happen.

The probability $\Pr(\mbox{Any two defective items have the same neighbor})$ equals $\frac{1}{S}$.
Then, by the union bound, the probability that there exist two defective items that have the same neighbor is at most $\frac{{D \choose 2}}{S}<\frac{D^{2}}{S}$.
Again, we choose $S={\cal O}(poly(D))$ with degree larger than $3$
to complete the proof.
\hfill \end{proof}

\noindent $a)$ {\it Number of tests:}


The number of tests in the first stage is ${\cal O}(D\log D\log S)={\cal O}(D\log^{2}D)$
and the number of tests in the second stage is ${\cal O}(D\log N)$.
Overall, the number of tests required is ${\cal O}(D(\log N+\log^{2}D))$.

\noindent $b)$ {\it Decoding complexity:}


The decoding complexity in the first stage is ${\cal O}(D\log S+D\log^{2}D)={\cal O}(D\log^{2}D)$
and the decoding complexity in the second stage is ${\cal O}(D\log N)$.
Overall, the decoding complexity is ${\cal O}(D(\log N+\log^{2}D))$.

\noindent $c)$ {\it Error probability:}


There are three events we need to consider.

First, the error probability of constructing bipartite graph ${\cal G}$
with the properties we need is ${\cal O}(1/poly(D))$.

Second, the error probability of non-adaptive group testing algorithm
is ${\cal O}(1/poly(D))$.

Third, the error probability of any localization testing is ${\cal O}(1/N)$.
By applying union bound over $D$ ${\cal D}$-leaf nodes, the probability
of incorrect decoding is ${\cal O}(1/poly(D))$.

Therefore, the overall error probability of incorrect decoding scales as
${\cal O}(1/poly(D))$.

\section{Conclusion}\label{sec:conclusion}
In this work we consider three group testing algorithms, specifically for adaptive, nonadaptive, and two-stage adaptive scenarios. In each of these scenarios, we present the first algorithms whose computational complexity is nearly information-theoretically order-optimal. The number of tests required in our algorithms is also nearly information-theoretically order-optimal (by the same factor).

\section{Appendix}
\subsection{Mcdiarmid's Inequality}
Let $X_{1},\ldots,X_{m}$ be independent random variables all taking
values in the set ${\cal X}$. Further, let $f:{\cal X}^{m}\mapsto\mathbf{R}$
be a function of $X_{1},\ldots,X_{m}$ that satisfies $\forall i$,
$\forall x_{1},\ldots,x_{m},$ $x'_{i}\in{\cal X}$,

\begin{eqnarray*}
\left|f\left(x_{1},\ldots,x_{i},\ldots,x_{m}\right)-f\left(x_{1},\ldots,x'_{i},\ldots,x_{m}\right)\right| & \leq & c_{i}
\end{eqnarray*}

Then for all $\epsilon>0$,

\begin{eqnarray*}
\Pr\left(f-{\bf E}\left[f\right]\geq\epsilon\right) & \leq & \exp\left(\frac{-2\epsilon^{2}}{\Sigma_{i=1}^{m}c_{i}^{2}}\right)
\end{eqnarray*}
and
\begin{eqnarray*}
\Pr\left(f-{\bf E}\left[f\right]\leq-\epsilon\right) & \leq & \exp\left(\frac{-2\epsilon^{2}}{\Sigma_{i=1}^{m}c_{i}^{2}}\right)
\end{eqnarray*}
%
%
%
%
%
%
%
%

\bibliographystyle{IEEEtran}
\bibliography{mayank,sid,mohammad,sheng}

\end{document}